\documentclass[11pt]{amsart}
\usepackage{amssymb,amsmath,amsthm,amsfonts,mathrsfs}
\usepackage[hmargin=3cm,vmargin=3.5cm]{geometry}
\usepackage[dvipsnames,table,xcdraw]{xcolor}
\usepackage[colorlinks = true,
            linkcolor = Fuchsia,
            urlcolor  = ForestGreen,
            citecolor = WildStrawberry,
            anchorcolor = blue]{hyperref}
\usepackage{graphicx,psfrag}
\usepackage{amscd}
\usepackage[shellescape]{gmp}
\usepackage{stmaryrd}  
\usepackage[all,2cell]{xy} \UseAllTwocells \SilentMatrices
\usepackage{tikz-cd}
\usepackage[dvips]{epsfig}
\usepackage{MnSymbol}
\usepackage{comment}
\usepackage{enumitem}
\usepackage{kbordermatrix}
\usepackage[colorinlistoftodos]{todonotes}
\usepackage{cancel}

\renewcommand{\kbldelim}{(}
\renewcommand{\kbrdelim}{)}

\usetikzlibrary{arrows,automata}
\usetikzlibrary{decorations.markings}
\usetikzlibrary{decorations.pathmorphing}

\makeatletter
\def\@adminfootnotes{%
  \let\@makefnmark\relax  \let\@thefnmark\relax
  \ifx\@empty\@date\else \@footnotetext{\@setdate}\fi
  \ifx\@empty\@subjclass\else \@footnotetext{\@setsubjclass}\fi
  \ifx\@empty\@keywords\else \@footnotetext{\@setkeywords}\fi
  \ifx\@empty\thankses\else \@footnotetext{%
    \def\par{\let\par\@par}\@setthanks}%
  \fi
}
\makeatother

\newtheorem{theorem}{Theorem}[section]
\newtheorem{proposition}[theorem]{Proposition}
\newtheorem{prop}[theorem]{Proposition}
\newtheorem{cor}[theorem]{Corollary}

\theoremstyle{definition}

\theoremstyle{remark}
\newtheorem{remark}[theorem]{Remark}

\numberwithin{equation}{section}

\let\oldtocsection=\tocsection
\let\oldtocsubsection=\tocsubsection
\renewcommand{\tocsection}[2]{\hspace{0em}\oldtocsection{#1}{#2}}
\renewcommand{\tocsubsection}[2]{\hspace{1em}\oldtocsubsection{#1}{#2}}

\usepackage{chngcntr}
\counterwithin{figure}{subsection}

\makeatletter
\@namedef{subjclassname@2020}{%
  \textup{2020} Mathematics Subject Classification}

\makeatother
 
\title[Topological quantum field theories over Dedekind domains]{Two-dimensional topological quantum field theories of rank two over Dedekind domains}

\author{Fabian Espinoza}
\address{Department of Mathematics, Johns Hopkins University, Baltimore, MD 21218, USA}
\email{\href{mailto:fespino5@jhu.edu}{fespino5@jhu.edu}}
 
\author{Mee Seong Im}
\address{Department of Mathematics, Johns Hopkins University, Baltimore, MD 21218, USA}
\email{\href{mailto:meeseong@jhu.edu}{meeseong@jhu.edu}}
 
\author{Mikhail Khovanov} 
\address{Department of Mathematics, Johns Hopkins University, Baltimore, MD 21218, USA}
\email{\href{mailto:khovanov@jhu.edu}{khovanov@jhu.edu}}

\date{February 28, 2025}

\providecommand{\keywords}[1]{\textbf{\textit{Key words and phrases.}} #1}


\renewcommand{\arraystretch}{1.00}

\begin{document}

\def\Aff{\mathsf{Aff}}
\def\AND{\mathsf{AND}}
\def\concatenate{\mathsf{concatenate}}
\def\Br{\mathsf{Br}}
\def\Gal{\mathsf{Gal}}
\def\gen{\mathsf{generators}}
\def\sl{\mathfrak{sl}}
\def\GL{\mathsf{GL}}
\def\SL{\mathsf{SL}}
\def\init{\mathsf{in}}
\def\t{\mathsf{t}}
\def\out{\mathsf{out}}
\def\I{\mathsf I}
\def\region{\mathsf{region}}
\def\PMI{\mathsf{PMI}}
\def\plane{\mathsf{plane}}
\def\R{\mathbb R}
\def\Q{\mathbb Q}
\def\Z{\mathbb Z}
\def\mc{\mathcal{c}}
\def\finite{\mathsf{finite}}
\def\infinite{\mathsf{infinite}}
\def\N{\mathbb N} 
\def\C{\mathbb C}
\def\sep{\mathsf{sep}}
\def\S{\mathbb S}
\def\SS{\mathbb S} 
\def\CP{\mathbb P}
\def\Ob{\mathsf{Ob}}
\def\op{\mathsf{op}}
\def\new{\mathsf{new}}
\def\old{\mathsf{old}}
\def\OR{\mathsf{OR}}
\def\AND{\mathsf{AND}}
\def\rat{\mathsf{rat}}
\def\rec{\mathsf{rec}}
\def\tail{\mathsf{tail}}
\def\coev{\mathsf{coev}}
\def\ev{\mathsf{ev}}
\def\id{\mathsf{id}}
\def\s{\mathsf{s}}
\def\S{\mathsf{S}}
\def\t{\mathsf{t}}
\def\start{\textsf{starting}}
\def\Notation{\textsf{Notation}}
\def\circleft{\raisebox{-.18ex}{\scalebox{1}[2.25]{\rotatebox[origin=c]{180}{$\curvearrowright$}}}}
\renewcommand\SS{\ensuremath{\mathbb{S}}}
\def\FF{\mathbf{F}}

\newcommand{\aver}{\mathsf{av}}  
\newcommand{\ophana}{\overline{\phantom{a}}}
\newcommand{\Bool}{\mathbb{B}}
\newcommand{\dmod}{\mathsf{-mod}}
\newcommand{\lang}{\mathsf{lang}}
\newcommand{\pfmod}{\mathsf{-pfmod}}
\newcommand{\primitive}{\mathsf{irr}} 
\newcommand{\Bmodo}[1]{\Bool_{#1}\mathsf{-mod}}  

\newcommand{\PP}{\mathcal{P}} 

\newcommand{\whA}{\widehat{A}}
\newcommand{\whC}{\widehat{C}}
\newcommand{\whM}{\widehat{M}}
\newcommand{\ClK}{\mathsf{Cl}(K)}
\newcommand{\Cl}{\mathsf{Cl}}

\newcommand{\mathT}{\mathsf{T}}
\newcommand{\mathF}{\mathsf{F}}
\newcommand{\mcS}{\mathcal{S}}
\newcommand{\mcN}{\mathcal{N}}
\newcommand{\mcF}{\mathcal F}
\newcommand{\mcL}{\mathcal{L}}
\newcommand{\wvareps}{\widetilde{\varepsilon}}
\newcommand{\undeps}{\underline{\varepsilon}}
\newcommand{\unda}{\underline{a}}
\newcommand{\undb}{\underline{b}}
\newcommand{\undt}{\underline{t}}
\newcommand{\undd}{\underline{d}}
\newcommand{\undcprime}{\underline{c'}}

\let\oldemptyset\emptyset
\let\emptyset\varnothing

\newcommand{\undempty}{\underline{\emptyset}}
\newcommand{\undsigma}{\underline{\sigma}}
\newcommand{\undtau}{\underline{\tau}}
\def\basis{\mathsf{basis}}
\def\irr{\mathsf{irr}} 
\def\spanning{\mathsf{spanning}}
\def\elmt{\mathsf{elmt}}

\def\H{\mathsf{H}}
\def\I{\mathsf{I}}
\def\II{\mathsf{II}}
\def\l{\lbrace}
\def\r{\rbrace}
\def\o{\otimes}
\def\lra{\longrightarrow}
\def\Ext{\mathsf{Ext}}
\def\mf{\mathfrak} 
\def\mcC{\mathcal{C}}
\def\mcU{\mathcal{U}}
\def\mcT{\mathcal{T}}
\def\mcO{\mathcal{O}}
\def\mcE{\mathcal{E}}
\def\Fr{\mathsf{Fr}}

\def\ovb{\overline{b}}
\def\tr{{\sf tr}} 
\def\str{{\sf str}} 
\def\det{{\sf det }} 
\def\tral{\tr_{\alpha}}
\def\one{\mathbf{1}}   

\def\lra{\longrightarrow}
\def\twoheadlra{\longrightarrow\hspace{-4.6mm}\longrightarrow}
\def\hooklra{\raisebox{.2ex}{$\subset$}\!\!\!\raisebox{-0.21ex}{$\longrightarrow$}}
\def\kk{\mathbf{k}}  
\def\gdim{\mathsf{gdim}}  
\def\rk{\mathsf{rk}}
\def\undep{\underline{\varepsilon}}
\def\mathM{\mathbf{M}}  

\def\Cob{\mathsf{Cob}} 
\def\Kar{\mathsf{Kar}}   

\newcommand{\mfp}{\mathfrak{p}}
\newcommand{\brak}[1]{\ensuremath{\left\langle #1\right\rangle}}
\newcommand{\brakspace}[1]{\ensuremath{\left\langle \:\: #1\right\rangle}}

\newcommand{\oplusop}[1]{{\mathop{\oplus}\limits_{#1}}}
\newcommand{\ang}[1]{\langle #1 \rangle } 
\newcommand{\ppartial}[1]{\frac{\partial}{\partial #1}} 
\newcommand{\checkr}{{\bf \color{red} CHECK IT}}
\newcommand{\checkb}{{\bf \color{blue} CHECK IT}}
\newcommand{\checkk}[1]{{\bf \color{red} #1}}

\newcommand{\mcA}{{\mathcal A}}
\newcommand{\cZ}{{\mathcal Z}}
\newcommand{\sq}{$\square$}
\newcommand{\bi}{\bar \imath}
\newcommand{\bj}{\bar \jmath}
\newcommand{\FinProb}{\mathsf{FinProb}}

\newcommand{\undn}{\underline{n}}
\newcommand{\undm}{\underline{m}}
\newcommand{\undzero}{\underline{0}}
\newcommand{\undone}{\underline{1}}
\newcommand{\undtwo}{\underline{2}}
\newcommand{\wtdelta}{\widetilde{\Delta}}

\newcommand{\cob}{\mathsf{cob}} 
\newcommand{\comp}{\mathsf{comp}} 

\newcommand{\Aut}{\mathsf{Aut}}
\newcommand{\Hom}{\mathsf{Hom}}
\newcommand{\Idem}{\mathsf{Idem}}
\newcommand{\Ind}{\mbox{Ind}}
\newcommand{\Id}{\textsf{Id}}
\newcommand{\End}{\mathsf{End}}
\newcommand{\iHom}{\underline{\mathsf{Hom}}}
\newcommand{\Bools}{\Bool^{\mathfrak{s}}}
\newcommand{\mfs}{\mathfrak{s}}
\newcommand{\blueline}{line width = 0.45mm, blue}

\newcommand{\drawing}[1]{
\begin{center}{\psfig{figure=fig/#1}}\end{center}}

\def\endomCempt{\End_{\mcC}(\emptyset_{n-1})}

\def\MS#1{\textbf{\color{NavyBlue}[MS: #1]}}
\def\MK#1{\textbf{\color{teal}[MK: #1]}}
\def\FE#1{\textbf{\color{Red}[FE: #1]}}

\begin{abstract}  
We give examples of Frobenius algebras of rank two over ground Dedekind rings which are projective but not free and discuss possible applications of these algebras to link homology.
\end{abstract}

\maketitle
\tableofcontents


%
%

\section{Introduction}
\label{section:intro}

By an $n$-dimensional topological quantum field theory (TQFT), we mean a monoidal functor $\mathcal{F}:\Cob_n\lra R\dmod$, where $\Cob_n$ is the category of $n$-dimensional cobordisms between $(n-1)$-dimensional manifolds, which we will write as $(n-1)$-manifolds, that are either oriented or unoriented, and $R$ is a commutative ring. The target category $R\dmod$ is the category of $R$-modules. 

It is easy to see that, for any $(n-1)$-manifold $N$, the module $\mathcal{F}(N)$ is a finitely-generated projective $R$-module; see~\cite[Proposition 2.1]{GIKKL23} for $n=1$, and the same argument works for any $n$. 

A one-dimensional oriented TQFT $\mathcal{F}$ is given by such a finitely generated projective $R$-module $P$. Namely, if $P=\mathcal{F}(+)$, then its dual $P^{\ast}\cong \mathcal{F}(-)$, where $+$ and $-$ are points viewed as 0-manifolds with positive and negative orientation, respectively. To \emph{cup} and \emph{cap} cobordisms~\cite{GIKKL23}, this TQFT assigns co-evaluation and evaluation maps, respectively. Consequently, the isomorphism classes of oriented one-dimensional TQFTs are in a bijection with finitely generated projective $R$-modules.

Topological quantum field theories are most commonly studied over a field $K$, with $\mcF(N)$ a finite-dimensional $K$-vector space for any $(n-1)$-manifold $N$.  At the same time, integral structures are often essential in a more refined analysis of TQFTs.   For example,  the Witten--Reshetikhin--Turaev (WRT) 3D TQFT~\cite{Witt89,ReshTur91} is defined over $\C$, with a root of unity parameter $q\in \C$. Integral structures in these TQFTs, see~\cite{CL05} and references therein, relate to 
Lusztig--Kashiwara canonical bases in irreducible representations of quantum groups~\cite{Kash90_crystal,Kash91_crystal,Lusz90_canonical,Lusz10_intro_quantum}. The integral structure also leads to the Habiro rings~\cite{Hab02,Hab04,GSWZ24}  and to the conjectural expansions at $q=0$ of the analytic extension of the WRT invariants into the unit disk. These expansions are expected to have integer coefficients which lift to the Euler characteristics of an unknown homology theory for 3-manifolds~\cite{GPV17}. In these developments, however, one does not study modifications of the WRT TQFT with state spaces of 2-manifolds being non-free modules over the ground ring. 

 Another notable example of integrality is the theory of invertible TQFTs. In an invertible $n$-dimensional TQFT, the object $\mcF(N)$ associated to an $(n-1)$-manifold $N$ is an invertible object of the target monoidal category, such as the category of line bundles over an algebraic variety, and the map that $\mcF$ associates to each $n$-cobordism is invertible~\cite{SchPri24,Freed14}. 

\vspace{0.07in} 

 In this paper, we start to investigate two-dimensional TQFTs over a commutative ring $R$ such that the state space $A$ of a circle is not a free $R$-module. A 2D TQFT is determined by the commutative Frobenius algebra structure on the $R$-algebra $A=\mathcal{F}(\SS^1)$ which is the state space of a circle $\SS^1$. 
If $R$ has no zero divisors (i.e. $R$ is a domain), consider the field of fractions $K=K(R)$. Tensoring with $K$ over $R$ converts commutative Frobenius pair $(A,R)$ to   commutative Frobenius pair $(A_K,K)$ over the field $K$, where  $A_K:=A\otimes_R K$.  

We are looking for examples of commutative Frobenius $R$-algebras $A$ which are projective but not free modules over $R$. 
For simplicity and convenience, assume that $R$ is a Dedekind ring. A Dedekind ring $\mcO$ often admit projective non-free modules, and the supply of the latter is measured by the ideal class group $
\Cl(\mcO)$.
 Examples of Dedekind rings include the ring of integers $\mcO_K$ in a number field $K$ and the ring of functions on a smooth irreducible algebraic curve.

Any projective $\mcO$-module $M$ of rank $n$ has the form $M\cong \mcO^{n-1}\oplus P$, where $P$ is a rank one projective $\mcO$-module. The isomorphism classes of the latter are parametrized by the elements of the ideal class group $\Cl(\mcO)$, which is an abelian group. If $\mcO=\mcO_K$ is the ring of integers in a number field $K$, the ideal class group is finite. For ideals $\mathfrak{a}$ and $\mathfrak{a}'$, viewed as projective $\mcO$-modules, we have $\mathfrak{a}\oplus \mathfrak{a}'\cong \mcO\oplus \mathfrak{a}\mathfrak{a}'$. 

\vspace{0.07in} 

Let $A$ be a commutative Frobenius algebra over a Dedekind ring $\mcO$. The algebra $A$ is equipped with the unit element, given by a morphism $\iota:\mcO\lra A$, with commutative, associative and unital multiplication $m:A^{\otimes 2}\lra A$ and with a nondegenerate trace map $\varepsilon:A\lra \mcO$. 

Define the dual $\mcO$-module $A^{\ast}:=\Hom_{\mcO}(A,\mcO)$. The map $\varepsilon$ as above induces a map of $\mcO$-modules 
\begin{equation} \label{eq_wvareps} 
{\wvareps}: A\lra A^{\ast}, \ \ A\ni a\mapsto \varepsilon_a\in A^{\ast}, \ \ \varepsilon_a(b):=\varepsilon(ab), \ \ b\in A.  
\end{equation}    
The map $\varepsilon$ is called \emph{nondegenerate} if $\widetilde{\varepsilon}$ is an isomorphism of $\mcO$-modules. The trace $\varepsilon$ is nondegenerate precisely when $A$ is Frobenius over $\mcO$.  

Dualize the multiplication to a map 
\[
A^{\ast}\stackrel{m^{\ast}}{\lra}(A\otimes_{\mcO} A)^{\ast}\cong A^{\ast}\otimes_{\mcO} A^{\ast}.
\]
Since $A$ is Frobenius,  $\varepsilon$ is nondegenerate, and 
 composing with $\wvareps$ and $\wvareps\,^{-1}\otimes \wvareps\,^{-1}$ results in the map 
$\Delta: A\lra A^{\otimes 2}$ dual to the multiplication $m$ in an appropriate sense:
\begin{gather}
    \begin{aligned}
\xymatrix@-1pc{
A^{*} \ar[rr]^{m^*} & & & & \hspace{-1.4cm} (A^{\otimes 2})^* \simeq   A^* \otimes A^* 
\ar@<-2ex>[dd]_{\wvareps^{-1}}^{\hspace{0.15cm}\otimes}  
\ar@<2.4ex>[dd]^{\wvareps^{-1}.}  \\  
& & & & \\ 
A \ar[uu]^{\wvareps} \ar@{.>}[rrrr]^{\Delta}
& & & & A\otimes A  \\ 
}
    \end{aligned}
\label{xymatrix_epsilontilde_maps}
\end{gather}

 
There is an $\mcO$-module decomposition $A\cong \mcO^{n-1}\oplus P$ for some projective rank one $\mcO$-module $P\in \Cl(\mcO)$. Then $A^{\ast}\cong \mcO^{n-1}\oplus P^{\ast}$. An isomorphism of $\mcO$-modules $A\cong A^{\ast}$ implies that $P\cong P^{\ast}$, so that $P\otimes P\cong \mcO$. 

\begin{cor} For a commutative Frobenius $\mcO$-algebra $A$ as above, 
    rank one projective $\mcO$-module $P$ has order at most two in the ideal class group $\Cl(\mcO)$. 
\end{cor}

In our search for projective non-free commutative Frobenius algebras $A$, we are thus limited to the case when $A\cong \mcO^{n-1}\oplus P$ with $P\otimes P\cong \mcO$ and $P\not\cong \mcO$. Thus, $P$ is a rank one projective of order two in the ideal class group. 

\vspace{0.07in} 

\noindent 
\textbf{Acknowledgments.} This paper resulted from a reading and research course given to one of the authors by the other two in the fall of 2024 at Johns Hopkins University. M.K.~would like to acknowledge partial support from NSF grant DMS-2204033 and Simons Collaboration Award 994328. 

%
%

\section{Rank two TQFTs}


\subsection{Integrality conditions and equations}
We now look for rank two commutative Frobenius algebras $A$ over a Dedekind ring $\mcO$ with the field of fractions $K=K(\mcO)$ such that $A$ is not a free $\mcO$-module. 

Since $\rk_{\mcO}(A)=2$,  we can decompose $A\cong \mcO\oplus \mu$, where projective module (or ideal) $\mu$ has order two in $\ClK$, $\mu\not\cong \mcO$, $\mu^2\cong \mcO$. We can assume that $\mu\subset \mcO$, write $\mu\mu=(z)$ for some $ z\in \mcO$, and interchangeably treat $\mu$ as a projective $\mcO$-module and as an ideal of $\mcO$. Assume that 
\begin{equation}\label{eq_A} 
A\cong \mcO\, 1 \oplus \mu \, X,
\end{equation}
where $1$ and $X$ are generators, and the multiplication is given by 
\begin{equation}\label{eq_mult}
X^2 = a X+b, \ \ a,b\in K. 
\end{equation}
Computing the multiplication on $\mu$-multiples of $X$, 
\[
(u_1X)(u_2X) = u_1u_2 aX + u_1 u_2 b, \ \ u_1,u_2\in \mu, 
\]
tells us that $u_1u_2a\in \mu$ and $u_1u_2b\in \mcO$ for all $u_1,u_2\in \mu$. Consequently, 
\[
a\in \mu^{-1}=z^{-1}\mu, \ \ b\in z^{-1}\mcO, 
\]
and any such pair $(a,b)$ turns $A$ into a commutative unital $\mcO$-algebra with the above multiplication. 

We next ask for a nondegenerate trace map $\varepsilon:A\lra \mcO$. The maps $\varepsilon$ that are $\mcO$-linear are described by the following data: 
\[
\varepsilon(1)\in \mcO, \ \ \varepsilon|_{\mu X} \ :\ \mu \lra \mcO.
\]
The map $\varepsilon: \mu X\lra \mcO$ is given by the element $\varepsilon(X)\in \mu^{-1}=z^{-1}\mu$. Thus, 
without the nondegeneracy condition, a trace map is given by a pair $(\varepsilon(1),\varepsilon(X))$ with $\varepsilon(1)
\in\mcO$ and $\varepsilon(X)\in z^{-1}\mu$. 

Together, the data of an algebra $A$ with a trace $\varepsilon$ is given by the quadruple 
\begin{equation}\label{eq_integrality_cond}
(a,b,\varepsilon(1),\varepsilon(X)), \ \ \  
a\in z^{-1}\mu, \ \ b\in z^{-1}\mcO, \ \ \varepsilon(1)\in \mcO, \ \ \varepsilon(X)\in z^{-1}\mu. 
\end{equation}

Next, consider the nondegeneracy condition on the trace. First, write the map $\wvareps$ in \eqref{eq_wvareps} explicitly over $K$, using the dual basis $\{1^{\ast},X^{\ast}\}$ for $A_K^{\ast}:=A^{\ast}\otimes_{\mcO}K$. Here 
\[
1^{\ast}(1)=1, \ \ 1^{\ast}(X)=0, \ \   X^{\ast}(1)=0, \ \  X^{\ast}(X)=1. 
\]
Integrally, 
\begin{equation}\label{eq_integrally}
    A\cong \mcO\, 1 \oplus \mu \, X\ \ \textrm{and} \ \   A^{\ast}=\mcO 1^{\ast}\oplus z^{-1}\mu X^{\ast}.
\end{equation}
The map 
\begin{equation}\label{eq_wvareps_a}
\widetilde{\varepsilon}:A\lra A^{\ast}
\end{equation} 
in \eqref{eq_wvareps} should be an isomorphism. 
 Let us compute this map, first over the field of fractions $K$:  
\begin{equation}\label{eq_tildeeps}
\wvareps(1)=\varepsilon(1)1^{\ast}+\varepsilon(X)X^{\ast}, \ \ 
\wvareps(X) = \varepsilon(X) 1^{\ast}+ \varepsilon(X^2) X^{\ast}=\varepsilon(X) 1^{\ast}+ (a\varepsilon(X)+b\varepsilon(1)) X^{\ast},  
\end{equation}
or, in matrix notation, 
\begin{equation}\label{eq_tildeeps_mat}
    \begin{pmatrix} \wvareps(1) \\  \wvareps(X) \end{pmatrix}  = 
    \begin{pmatrix} \varepsilon(1) & \varepsilon(X) \\
    \varepsilon(X) & \varepsilon(X^2)
    \end{pmatrix}
    \begin{pmatrix} 1^{\ast} \\  X^{\ast}\end{pmatrix}.
\end{equation}
Let 
\begin{equation}\label{eq_M_wtdelta}
    M:= \begin{pmatrix} \varepsilon(1) & \varepsilon(X) \\ \varepsilon(X) & \varepsilon(X^2) \end{pmatrix}, \ \ 
    \wtdelta := \det(M) = \varepsilon(1)\varepsilon(X^2)-\varepsilon(X)^2. 
\end{equation}
Necessarily, $\wtdelta\not= 0$. 
The inverse map $\wvareps\,^{-1}:A^{\ast}_K\lra A_K$ in these bases is given by the inverse matrix $M^{-1}$: 
\begin{equation}\label{eq_eps_inv}
    \begin{pmatrix} \wvareps\,^{-1}(1^{\ast}) \\  \wvareps\,^{-1}(X^{\ast}) \end{pmatrix}  = \wtdelta^{-1}
    \begin{pmatrix} \varepsilon(X^2) & -\varepsilon(X) \\
    -\varepsilon(X) & \varepsilon(1)
    \end{pmatrix}
    \begin{pmatrix} 1 \\  X\end{pmatrix}.
\end{equation}
Let us write down multiplication in $A_K$: 
\begin{equation}
   m: \  1\otimes 1 \mapsto 1, \ \  \ 
   1\otimes X\mapsto X, \ \ \ 
   X\otimes 1 \mapsto X, \ \ \  
   X\otimes X \mapsto aX + b1,
\end{equation}
and transpose it to comultiplication in $A_K^{\ast}$:
\begin{equation}
   m^{\ast}: \ 1^{\ast}\lra 1^{\ast}\otimes 1^{\ast} + b X^{\ast}\otimes X^{\ast}, \ \ \ 
   X^{\ast}\lra 1^{\ast}\otimes X^{\ast}+X^{\ast}\otimes 1^{\ast} + a X^{\ast}\otimes X^{\ast}. 
\end{equation}
We have 
\[
\Delta = (\wvareps\,^{-1} \otimes \wvareps\,^{-1})\circ m^{\ast}\circ \wvareps
\]
and 
\begin{eqnarray*}
    \Delta(1) & = & (\wvareps\,^{-1} \otimes \wvareps\,^{-1})\circ m^{\ast}\circ \wvareps (1) \\
    & = & (\wvareps\,^{-1} \otimes \wvareps\,^{-1})\circ m^{\ast} (\varepsilon(1) 1^{\ast}+\varepsilon(X)X^{\ast}) \\
    & = & (\wvareps\,^{-1} \otimes \wvareps\,^{-1}) (\varepsilon(1) \left(1^{\ast}\otimes 1^{\ast} + b X^{\ast}\otimes X^{\ast}\right) + \varepsilon(X)\left(1^{\ast}\otimes X^{\ast}+X^{\ast}\otimes 1^{\ast} + a X^{\ast}\otimes X^{\ast}\right) \\ 
    & = & (\wvareps\,^{-1} \otimes \wvareps\,^{-1})\left(\varepsilon(1)1^{\ast}+\varepsilon(X)X^{\ast}\right)\otimes 1^{\ast} +\left(\varepsilon(X)1^{\ast}+\varepsilon(X^2)
    X^{\ast}\right)\otimes X^{\ast}.
\end{eqnarray*}
We compute
\[
\wvareps\,^{-1} (\varepsilon(1)1^{\ast}+\varepsilon(X)X^{\ast}) = \wtdelta^{-1}(\varepsilon(1)(\varepsilon(X^2)1-\varepsilon(X)X)+\varepsilon(X)(-\varepsilon(X)1+\varepsilon(1)X)) = 1, 
\]
and 
\[
\wvareps\,^{-1} (\varepsilon(X)1^{\ast}+\varepsilon(X^2)
    X^{\ast}) = \wtdelta^{-1}\left(\varepsilon(X)(\varepsilon(X^2)1-\varepsilon(X)X) + \varepsilon(X^2) (-\varepsilon(X)1+\varepsilon(1)X)\right) = X. 
\]
Thus, 
\[
\Delta(1) = 1 \otimes \wvareps\,^{-1}(1^{\ast}) + X \otimes \wvareps\,^{-1}(X^{\ast})= \wtdelta^{-1}\left( 1\otimes (\varepsilon(X^2)1-\varepsilon(X)X ) + X \otimes (- \varepsilon(X)1 + \varepsilon(1)X)
\right) 
\]
and we arrive at the formula
\begin{equation}\label{eq_delta_one}
\Delta(1) = \wtdelta^{-1}\left( 
\varepsilon(X^2)1\otimes 1 -\varepsilon(X)(1\otimes X+X\otimes 1)+\varepsilon(1)X\otimes X\right) 
\end{equation}
in the Frobenius $K$-algebra $A_K$. Rewrite this formula as 
\begin{equation}\label{eq_delta2_one}
\Delta(1) = \wtdelta^{-1}\left( 
1\otimes(\varepsilon(X^2) 1-\varepsilon(X)X) +X\otimes (-\varepsilon(X)1 +\varepsilon(1) X)\right).
\end{equation}
It tells us that the dual to the basis $(1,X)$ of $A_K$ is the basis 
\begin{equation}
    \wtdelta^{-1}\left(\varepsilon(X^2) 1-\varepsilon(X)X, -\varepsilon(X)1 +\varepsilon(1) X\right). 
\end{equation}

Map $\wvareps\,^{-1}:A_K^{\ast}\lra A_K$ on $K$-vector spaces should restrict to a map $\wvareps\,^{-1}:A^{\ast}\lra A$. From \eqref{eq_integrally}  this means 
\[
\wvareps\,^{-1}(1^{\ast}), \wvareps\,^{-1}(z^{-1}u X^{\ast})\in \mcO 1 \oplus \mu X, \ \forall u \in \mu. \ 
\]
Equation \eqref{eq_eps_inv} gives 
\begin{equation}
\wvareps\,^{-1}(1^{\ast}) = \wtdelta^{-1}\left(\varepsilon(X^2)1 - \varepsilon(X)X\right), \ \ 
\wvareps\,^{-1}(z^{-1}u X^{\ast}) = \wtdelta^{-1}z^{-1}u \left(-\varepsilon(X)1 + \varepsilon(1)X\right)
\end{equation}
and the integrality conditions are 
\begin{equation}\label{eq_int_cond}
    \wtdelta^{-1}\varepsilon(X^2)\in \mcO, \ \ \wtdelta^{-1}\varepsilon(X)\in \mu, \ \ \wtdelta^{-1}z^{-1}u \varepsilon(X)\in \mcO, \ \ \wtdelta^{-1}z^{-1}u \varepsilon(1)\in \mu, \ \ \  \forall u\in \mu, 
\end{equation}
where $\wtdelta$ is given by \eqref{eq_M_wtdelta}. These conditions can be  rewriten as 
\begin{equation}\label{eq_three_conditions}
    \varepsilon(X^2)\in \wtdelta\,\mcO, \ \ 
    \varepsilon(X) \in \wtdelta\,\mu, \  \  
    \varepsilon(1)\in \wtdelta z \,\mcO. 
\end{equation}

Let us go back to formulas \eqref{eq_integrally} and \eqref{eq_tildeeps} and continue with the integrality constraints. 
Integrally, $\wvareps(1)\in A^{\ast}$, and for $\wvareps(uX)\in A^{\ast}$, $u\in \mu$, one additionally needs 
\[
\varepsilon(u X) = u \varepsilon(X) 1^{\ast}+ u(a\varepsilon(X)+b\varepsilon(1))X^{\ast}\in \mcO 1^{\ast}\oplus z^{-1}\mu X^{\ast}, \ \ u\in \mu. 
\]
Note that $\varepsilon(X^2)=a\varepsilon(X)+b\varepsilon(1)$.
Equivalently, we need $u \varepsilon( X) \in \mcO$, which holds due to \eqref{eq_integrality_cond}, and $u(a\varepsilon(X)+b\varepsilon(1))\in z^{-1}\mu$ since 
\[u a \varepsilon(X)\in \mu z^{-1}\mu z^{-1}\mu = z^{-1}\mu 
\]
and 
\[
u b \,\varepsilon(1) \in \mu z^{-1}\mcO  = z^{-1}\mu .
\] 
Next, for $\wvareps:A\lra A^{\ast}$ to be an isomorphism, one needs $1^{\ast}$ and $z^{-1}u X^{\ast}$ for $u\in \mu$ to be in the image of $\wvareps$. 

\vspace{0.1in} 

I. For $1^{\ast}$ to be in the image of $\wvareps$, one needs $\wvareps(c1+dX)=1^{\ast}$ for some $c$ and $d$. 
That is,
\begin{align*}
\varepsilon_{c1+dX}(1) = \varepsilon (c1+dX) = c\varepsilon(1)+d\varepsilon (X) = 1
\end{align*}
and
\begin{align*}
\varepsilon_{c1+dX}(X) &= \varepsilon((c1+dX)X) = \varepsilon(cX+dX^2) = \varepsilon (cX+d(aX+b)) \\ 
&= c\varepsilon(X) + ad \varepsilon(X) + bd \varepsilon(1) = 0.
\end{align*} 
Thus, for some $c\in \mcO, d \in \mu$, the two equations 
\begin{equation}
c \varepsilon(1) + d \varepsilon(X)=1, 
\ \ 
c\varepsilon(X) + d(a \varepsilon(X)+ b \varepsilon(1))=0 
\end{equation}
should hold, or, in matrix form, they are written as 
\begin{equation}\label{eq_A1_correct}
\begin{pmatrix} c & d \\ bd  & ad+c 
\end{pmatrix} \, \begin{pmatrix} \varepsilon(1) \\  \varepsilon(X)\end{pmatrix} = \begin{pmatrix} 1 \\ 0 
\end{pmatrix}.
\end{equation}
A useful parameter is 
$t=a \varepsilon(X)+b\varepsilon(1)=\varepsilon(X^2)$, also see \eqref{eqn_parameter_t}. 

\vspace{0.1in} 

II. For $z^{-1}u X^{\ast}$ to be in the image of $\wvareps$, where $u\in \mu$, one needs $c'(u)=u c'(1)\in \mcO$ and $d'(u)=u d'(1)\in \mu$, where $\wvareps(c'(u)+d'(u)X)= z^{-1}u X^{\ast}$. Here, $c':\mu\lra\mcO$ and $d':\mu\lra \mu$ are suitable $\mcO$-linear maps that describe elements mapped to multiples of $X^{\ast}$ by $\wvareps$. These extend to $\mcO$-linear maps $c':\mcO\lra z^{-1}\mu$ and $d':\mcO\lra \mcO$.

Similar to Case I, we have 
\begin{align*}
\varepsilon_{c'(u)+d'(u)X}(1) 
= \varepsilon(c'(u)+d'(u)X) 
= c'(u) \varepsilon(1) + d'(u)\varepsilon(X) = 0,
\end{align*}
and 
\begin{align*}
\varepsilon_{c'(u)+d'(u)X}(X)
&= \varepsilon(c'(u)X+d'(u)X^2) 
= \varepsilon(c'(u)X+d'(u)(aX+b)) \\
&= c'(u)\varepsilon(X) + ad'(u)\varepsilon(X) + bd'(u)\varepsilon(1) = uz^{-1}.
\end{align*}
These relations give us  
\begin{equation}\label{eq_A3_correct}
\begin{pmatrix} c'(u) & d'(u) \\ b d'(u)  & a d'(u)+c'(u)
\end{pmatrix} \, \begin{pmatrix} \varepsilon(1) \\  \varepsilon(X)\end{pmatrix} = \begin{pmatrix} 0 \\ u z^{-1}
\end{pmatrix}, 
\end{equation} 
Since $c'(u)$ and $d'(u)$ are linear in $u$, we can factor $u$ out and use shorthand notation 
\begin{equation}\label{eq_cdprime}
    c':=c'(1)\in z^{-1}\mu, \ \  \  d':=d'(1)\in\mcO
\end{equation} 
to get 
\begin{equation}\label{eq_A_correct}
\begin{pmatrix} c' & d' \\ b d'  & a d'+c'
\end{pmatrix} \, \begin{pmatrix} \varepsilon(1) \\  \varepsilon(X)\end{pmatrix} = \begin{pmatrix} 0 \\ z^{-1}
\end{pmatrix}. 
\end{equation}
Write this as 
\[
c' \varepsilon(1) + d' \varepsilon(X)=0, 
\ \ 
c'\varepsilon(X) + d'(a \varepsilon(X)+ b\varepsilon(1))=z^{-1}. 
\]

 {\bf Summary I:} {\it Given  $a,b,\varepsilon(1),\varepsilon(X)$ as in \eqref{eq_integrality_cond} and copied below:
\begin{equation}\label{eq_integrality_cond_copy_1}
(a,b,\varepsilon(1),\varepsilon(X)), \ \ \  
a\in z^{-1}\mu, \  \ 
b\in z^{-1}\mcO, \ \ 
\varepsilon(1)\in \mcO, \ \ 
\varepsilon(X)\in z^{-1}\mu, 
\end{equation}
also define 
\begin{equation}
\label{eqn_parameter_t}
t := a \varepsilon(X)+b\varepsilon(1)=\varepsilon(X^2),
\end{equation}
where the second equality is obtained by applying $\varepsilon$ to \eqref{eq_mult}
(in particular, $t\in z^{-1}\mcO$),
and look for $c\in \mcO,d\in \mu,c'\in z^{-1}\mu,d'\in \mcO$ such that matrix relations \eqref{eq_A1_correct} and \eqref{eq_A_correct} hold. These are linear relations 
\begin{eqnarray}\label{eq_four_one_corr}
   c \varepsilon(1)+d \varepsilon(X) & = & 1, \\
   \label{eq_four_two_corr}
   c\varepsilon(X) + d t& = & 0 , \\
   \label{eq_four_three_corr}
    c' \varepsilon(1) + d' \varepsilon(X) & = & 0, \\
    \label{eq_four_four_corr}
    c' \varepsilon(X)+d' t & = & z^{-1},
\end{eqnarray}
which can also be written as matrix equations (cf.~\eqref{eq_M_wtdelta})
\begin{equation}\label{eq_2-by-2}
   \begin{pmatrix} \varepsilon(1) & \varepsilon(X) \\ \varepsilon(X) & \varepsilon(X^2) \end{pmatrix}
    \begin{pmatrix} c \\ d \end{pmatrix} = 
    \begin{pmatrix} 1 \\ 0 \end{pmatrix}, \ \  \begin{pmatrix} \varepsilon(1) & \varepsilon(X) \\ \varepsilon(X) & \varepsilon(X^2) \end{pmatrix}
    \begin{pmatrix} c' \\ d' \end{pmatrix} = 
    \begin{pmatrix} 0 \\ z^{-1} \end{pmatrix}.
\end{equation}
}

\begin{center}
\begin{tabular}{ |c|c|c|c| } 
\hline
\multicolumn{4}{|c|}{Integrality conditions on the parameters} \\
\hline 
$z^{-1}\mcO$ & $z^{-1}\mu$ & $\mcO$ & $\mu$ \\
\hline
\hline
$b$ &  $a$  &   &  \\ 
    & $\varepsilon(X)$ & $\varepsilon(1)$ &   \\  
\hline 
 &   & $c$ & $d$ \\ 
& $c'$ & $d'$ &  \\ 
\hline
$t$ & & & \\ 
\hline
\end{tabular}
\end{center}
In the table above, the ideals decrease from left to right: $z^{-1}\mcO\supset z^{-1}\mu\supset \mcO\supset \mu$. The four columns describe where corresponding parameters should lie. We have different rows for the two different quadruples of parameters. In the relations above, $c$ and $d$ interact only with the first four parameters (before the change of variables using $t$), and likewise for $c'$ and $d'$. 

Multiply \eqref{eq_four_one_corr} by $\varepsilon(X)$ and subtract from it   \eqref{eq_four_two_corr} multiplied by $\varepsilon(1)$ to obtain 
\begin{equation}\label{eq_d_quad}
    d\left(\varepsilon(X)^2-\varepsilon(1)\varepsilon(X^2)\right)= \varepsilon(X).
\end{equation}
(Recall that $\varepsilon(X)^2-\varepsilon(1)\varepsilon(X^2)=-\det(M)=-\wtdelta$, see \eqref{eq_M_wtdelta}.) 
Next, write 
\eqref{eq_four_three_corr} as 
$\varepsilon(X)d'=-c'\varepsilon(1)$ and insert into \eqref{eq_four_four_corr} multiplied by $\varepsilon(X)$ to get 
\begin{equation}\label{eq_cprime_quad}
c'\left(\varepsilon(X)^2-\varepsilon(1)\varepsilon(X^2)\right)= z^{-1}\varepsilon(X).
\end{equation}
From \eqref{eq_d_quad} and \eqref{eq_cprime_quad}, we obtain
\begin{equation}
    (d-zc')\left(\varepsilon(X)^2-\varepsilon(1)\varepsilon(X^2)\right)=0. 
\end{equation}
Consequently, either (I) $d=zc'$ or (II) $\varepsilon(X)^2=\varepsilon(1)\varepsilon(X^2)$.
If (II) holds, then the $2\times 2$ matrix in \eqref{eq_2-by-2} has rank 1 over the field of fractions $K$, and the system of equations in \eqref{eq_2-by-2} cannot both hold. Hence, 
\begin{equation}\label{eq_d_cprime}
 d \ = \ zc'
\end{equation}
and 
\begin{equation}\label{eq_t_not_zero}
 -\wtdelta=\varepsilon(X)^2-\varepsilon(1)\varepsilon(X^2) \not=  0.
\end{equation}
Now write 
\eqref{eq_four_three_corr} as 
\begin{equation}
\label{eqn_dprimez_two}
d \varepsilon(1) + d'z \varepsilon(X)=0
\end{equation}
and 
\eqref{eq_four_four_corr} as
\begin{equation}
\label{eqn_dprimez}
d \varepsilon(X)+ d' z\varepsilon(X^2)  = 1.
\end{equation}
Rescaling \eqref{eqn_dprimez} by $\varepsilon(1)$ and \eqref{eqn_dprimez_two} by $\varepsilon(X)$ 
and then subtracting one from the other gives 
\[
d'z \left(\varepsilon(X^2)\varepsilon(1) -  \varepsilon(X)^2\right) = \varepsilon(1).
\]
That $d'$ defined by this equation is in $\mcO$, see the table above, is equivalent to the third condition in \eqref{eq_three_conditions}. 

We can express $c'$ and essentially reduce $d'$ to $d$ via the equations
\begin{equation}\label{eq_reduction_to_d}
c'=z^{-1}d, \ \ \  d'\varepsilon(X)z = - d \varepsilon(1) 
\end{equation}
since $z\in K$ for the former and using \eqref{eq_d_cprime} and \eqref{eq_four_three_corr} for the latter.

(i) If $\varepsilon(X)\not=0$, we can express both $c',d'$ via $d$, and equations \eqref{eq_four_three_corr} and \eqref{eq_four_four_corr} become redundant.  
Integrality condition on $c'$ reduces to that for $d$, integrality condition on $d'$ becomes 
\begin{equation}\label{eq_instead_cdprime}
   d \varepsilon(1) \varepsilon(X)^{-1}\in (z).
\end{equation}
We can then remove $c',d'$ from the list of our parameters, remove \eqref{eq_four_three_corr}, \eqref{eq_four_four_corr} and add integrality condition \eqref{eq_instead_cdprime}.

(ii) Case $\varepsilon(X)=0$ is more straightforward, implying $d=0$ and $c'=0$, via \eqref{eq_d_quad} and \eqref{eq_cprime_quad}, respectively. 

\vspace{0.07in} 

For now let us consider both cases at once and rescale the parameters to reduce them to elements of ideals $\mcO$ and $\mu$. 
Rescale the parameters as follows
\begin{equation}\label{eq_param_shift_new}
b=z^{-1} \undb, \ \  a = z^{-1}\unda, \ \
\varepsilon(X) = z^{-1}\undeps_X, 
\ \  t=z^{-1}\undt
\end{equation}
so that
\[
\undb,\undt\in \mcO, \ \ \ \ \unda,\undeps_X\in \mu, \ \
 d\: \varepsilon(1)\:  \in \undeps_X\mcO.
\]
The corresponding table is 
\begin{center}
\begin{tabular}{ |c|c|c| } 
\hline
\multicolumn{3}{|c|}{Integrality conditions on rescaled parameters} \\ 
\hline 
 $\mcO$ & $\mu$ & $\undeps_X\mcO$\\
\hline
\hline 
  $\undb$  & $\unda$ & \\ 
$\varepsilon(1)$ & $\undeps_X$ &  \\  
\hline 
  $c,d'$ & $d$ & \\ 
\hline
    $\underline{t}$ &  & $d\varepsilon(1)$ \\
    \hline 
\end{tabular}.
\end{center}
Now 
\begin{equation}\label{eq_undt}
\undt = z^{-1}\unda\: \undeps_X+\undb\varepsilon(1),
\end{equation}
and the relations \eqref{eq_four_one_corr}-\eqref{eq_four_four_corr} become 
\begin{eqnarray}\label{eq_four_one_z}
   c \:\varepsilon(1)+z^{-1}d\: \undeps_X & = & 1, \\
   \label{eq_four_two_z}
   c\:\undeps_X  & = & - d \:\undt , \\
   \label{eq_four_three_z}
   d'\undeps_X & = & - d\: \varepsilon(1) , \\
   \label{eq_four_four_z}
   z^{-1}\, d\, \undeps_X + d'\,\undt & = & 1. 
\end{eqnarray}
Subtracting \eqref{eq_four_one_z} from \eqref{eq_four_four_z}, we can replace the latter by 
\begin{equation}\label{eq_four_four_a}
    d'\, \undt  =  c \,\varepsilon(1).
\end{equation}
Equations~\eqref{eq_four_one_z}-\eqref{eq_four_four_z} 
can also be written in the following matrix form 
\begin{equation}\label{eq_new_matrix}
\begin{pmatrix} \varepsilon(1) & z^{-1}\undeps_X \\ \undeps_X  & \undt 
\end{pmatrix} \, \begin{pmatrix} c \\  d\end{pmatrix} = \begin{pmatrix} 1 \\ 0
\end{pmatrix}, \ \ \ \ 
\begin{pmatrix} \varepsilon(1) & \undeps_X \\ z^{-1}\undeps_X  & \undt 
\end{pmatrix} \, \begin{pmatrix} d \\  d'\end{pmatrix} = \begin{pmatrix} 0 \\ 1
\end{pmatrix}.
\end{equation}

While the multiplication in $A_K$ is given by \eqref{eq_mult}, the multiplication in $A$ is based on multiplying two elements in $\mu X$ and given by 
\begin{equation}
    \label{eq_mult_A}
    u_1 X \, u_2 X = \frac{u_1 u_2}{z}\,(\unda X +\undb), \ \ \  u_1,u_2\in \mu. 
\end{equation}
This shows another use for $\unda,\undb$. 
Note that $u_1u_2z^{-1},\undb\in\mcO$, $\unda\in \mu$, so the right hand side is in $A=\mcO 1\oplus \mu X$. 

\begin{prop}\label{prop_not_equal}
    In any example as above, $\varepsilon(1)\not=0$,  $\undt\not=0$, $c\not=0$, and $d'\not=0$.
\end{prop}
\begin{proof}
I. Assume $\varepsilon(1)=0$. Then $\undt=z^{-1}\unda \: \undeps_X$ and  equation \eqref{eq_four_one_z} becomes $d\undeps_X=z$, giving us equality of principal ideals $(d)(\undeps_X)=(d \undeps_X)=(z)$. There are proper inclusions of ideals $(d)\subset \mu$ and $(\undeps_X)\subset \mu$. These produce a chain of inclusions of ideals 
\[
(z)=(d)(\undeps_X) \subset (d)\mu  \subset   \mu\mu = (z),  
\]
implying that all ideals in the chain are equal. But $(z)$ is principal and $(d)\mu$ is not principal, giving a contradiction. Hence $\varepsilon(1)\not=0.$

\vspace{0.05in}

II. Assume $\undt=0$. We already know from \eqref{eq_t_not_zero} that $t=\varepsilon(X)^2-\varepsilon(1)\varepsilon(X^2)\not=0$, so that $\undt\not=0$. Here is another proof: from \eqref{eq_four_two_z},  
$c\,\undeps_X=0$. If $c=0$, from Equation \eqref{eq_four_one_z},  $d\undeps_X=z$. We get a contradiction, as in the proof of part I, since $d,\undeps_X\in \mu$.  Consequently, $\undeps_X=0$. Equation \eqref{eq_four_four_z} then implies $d'\undt=1$, which is a contradiction since we have assumed $\undt=0$.

\vspace{0.05in}
 
III. 
If $c=0$, then \eqref{eq_four_two_z} results in  $d=0$, giving a contradiction with \eqref{eq_four_one_z}, $0=1$. So $c\not=0$.

\vspace{0.05in}

IV. If $d'=0$, then \eqref{eq_four_four_z} reduces to $d\undeps_X=z$, leading to a contradiction as in part I. 
\end{proof}

\begin{remark}[Twistings] 
\label{rm_twistings}  
Given a 2D Frobenius algebra $(A,\mu,\varepsilon)$ as above, it admits the following modifications and parameter changes: 
\begin{enumerate}
    \item\label{item_invertibleLambdaNot} Change of variables $X\mapsto\lambda_0X$, given an invertible $\lambda_0\in \mcO^{\times}$: 
    \begin{equation}\label{eq_twist_one}
    a\mapsto \lambda_0a, \ \ \  
    b\mapsto \lambda_0^2b, \ \ \ 
    \varepsilon(X)\mapsto \lambda_0\varepsilon(X), \ \ \ 
    \varepsilon(1)\mapsto \varepsilon(1). 
    \end{equation} 
    \item\label{item_additive_LambdaOne} Change of variables $X\mapsto X+\lambda_1$, for $\lambda_1\in z^{-1}\mu$. This changes
    \begin{equation}\label{eq_twist_two}
    a \mapsto a+2\lambda_1, \ \ \ 
    b \mapsto b + \lambda_1^2.  \  \ 
    \end{equation}
    \item\label{item_invertibleLambda} For an invertible $\lambda\in \mcO^{\times}$ rescale 
    \begin{equation}\label{eq_twist_three}
    \varepsilon \mapsto \lambda\varepsilon,  \ \ \ 
    \Delta \mapsto \lambda^{-1}\Delta. 
    \end{equation}
\end{enumerate}
Invertible linear changes \eqref{item_invertibleLambdaNot} and \eqref{item_additive_LambdaOne} of $X$  together generate an action of the affine group $\Aff(\mcO,\mu)$ of symmetries ($X\mapsto aX+b$) on the set of Frobenius algebras as above with $X$ fixed. 
\end{remark}

%
%

\section{Some solutions}

Let us now consider two cases, depending on whether or not $\varepsilon(X)$ vanishes. 

\vspace{0.07in}


\subsection{Case \texorpdfstring{$\varepsilon(X)=0$}{epsilonXEQUALSzero}}
\label{subsec_caseone}


\begin{prop}\label{prop_undeps}
    Assume that in an example as above, $\undeps_X=0$. Then 
    \begin{equation}\label{eq_undepsX_zero}
    c'=d=0, \ \ \unda\in\mcO,\ \ \undb,\varepsilon(1)\in \mcO^{\times}, \ \ \undt=\undb\varepsilon(1)\in \mcO^{\times}, \ \ c=\varepsilon(1)^{-1}\in \mcO^{\times}, \ \ d'=\undb^{-1}\varepsilon(1)^{-1}\in \mcO^{\times}.
    \end{equation}
    Each of these data provides a commutative Frobenius algebra 
    \begin{equation}\label{eq_when_eXiszero}
    A=\mcO 1\oplus \mu X, \ \ 
    \mu^2=(z), \ \ X^2=z^{-1}\unda X+z^{-1}\undb, 
    \end{equation}
    uniquely determined by the ideal $\mu$, by $z$, by $\unda\in\mcO$ and $\undb,\varepsilon(1)\in\mcO^{\times}$ via the above formulas.
\end{prop}

\begin{proof}
Assume $\undeps_X=0$. Then $d\undt=0$ by \eqref{eq_four_two_z}. Since $\undt\not=0$, we get $d=0$. Then $c\varepsilon(1)=1$ by \eqref{eq_four_one_z}, so that $\varepsilon(1)\in\mcO^{\times}$ and $c=\varepsilon(1)^{-1}$. The last two equations~\eqref{eq_four_three_z} and \eqref{eq_four_four_z} give $\underline{c'}\varepsilon(1)=0$, $d'\undt=1$. Consequently, $\underline{c'}=0$ and $\undt\in\mcO^{\times}$, $d'=\undt^{-1}$. Then \eqref{eq_undt} gives $\undt=\undb\varepsilon(1)$. This requires $\undb\in \mcO^{\times}$. Summarizing, in this case: 
\[
\undeps_X=c'=d=0,\ \unda\in\mcO, \ \varepsilon(1),\undb\in \mcO^{\times}, \ \undt=\undb\varepsilon(1)\in \mcO^{\times}, \  c=\varepsilon(1)^{-1}\in \mcO^{\times}, \ d'=\undb^{-1}\varepsilon(1)^{-1}\in \mcO^{\times}. 
\]
This concludes the proof. 
\end{proof}
 
{\it 
 Thus, there is a family of solutions to the problem of constructing a Frobenius $\mcO$-algebra $A=\mcO 1\oplus \mu X$, with the square of the generator $X$ given by \eqref{eq_mult}, with a non-principal ideal $\mu\subset \mcO$ with $\mu^2=(z)$ and trace $\varepsilon$ as in \eqref{eq_integrality_cond}. With the additional condition that $\varepsilon(X)=0$, the solutions are classified by $\unda\in \mcO$ and $\varepsilon(1),\undb\in\mcO^{\times}$, with the parameters given by \eqref{eq_undepsX_zero} and \eqref{eq_param_shift_new}, see also \eqref{eq_when_eXiszero}. 
The trace $\varepsilon$ is $0$ on $\mu X$, and the restriction of $\varepsilon$ to $\mcO 1$ is an isomorphism $\mcO\lra\mcO$, $x\mapsto \varepsilon(1)x$. Isomorphism 
${\wvareps}: A\lra A^{\ast}$ in \eqref{eq_wvareps} is given by 
\[
\wvareps(1) = \varepsilon(1)1^{\ast}, \ \ \  \wvareps(uX) = z^{-1}\undb  \varepsilon(1)\, u X^{\ast},  \ \ u\in \mu. 
\]
This construction works for any Dedekind domain $\mcO$ with a non-principal ideal $\mu$ whose square is a principal ideal. 
 }

\begin{remark}
    The comultiplication map $\Delta:A\lra A\otimes A$ is dual to the multiplication  via the formula in~\eqref{xymatrix_epsilontilde_maps}. Since it is a map of $A$-bimodules, it is enough to compute $\Delta(1)$. One can first compute it in the field of fractions $K$  to get 
    \begin{equation}\label{eq_delta1_case0}
    \Delta(1) = \varepsilon(1)^{-1}(1\otimes 1+b^{-1}X\otimes X). 
    \end{equation}
    Also note that \eqref{eq_delta2_one} in case $\varepsilon(X)=0$ gives 
    \begin{eqnarray*}
    \Delta(1) & = & \wtdelta^{-1}\left( 
1\otimes\varepsilon(X^2) 1 +X\otimes \varepsilon(1) X\right) = \varepsilon(1)^{-1}1\otimes 1 + \varepsilon(X^2)^{-1}X\otimes X \\
 & = & 
\varepsilon(1)^{-1}\left( 1\otimes 1 + \varepsilon(1)\varepsilon(X^2)^{-1}X\otimes X\right).
    \end{eqnarray*}
    We have $b=\varepsilon(X^2)\varepsilon(1)^{-1}$ for this case. Indeed, from \eqref{eq_undepsX_zero} we get $\undt=\undb \varepsilon(1)$ so that $t=b\varepsilon(1)$ and $t=\varepsilon(X^2)$, resulting in the coefficient $b^{-1}$ next to $X\otimes X$ in \eqref{eq_delta1_case0}.  
    
    In the $\mcO$-algebra $A$ term $X$ is only available when scaled by elements of $\mu$. 
    Pick 2 generators $s_1,s_2$ of the ideal $\mu$, so that $\mu=\mcO s_1+\mcO s_2$. Since $\mu^2=z$, pick elements $s_1',s_2'\in \mu$ so that $s_1s_1'+s_2s_2'=z$. The element $z X\otimes X$ can be replaced by the sum of two product elements in $\mu X\otimes \mu X$ as shown below (and using that $b^{-1}=\undb^{-1}z$) 
    \begin{equation}\label{eq_delta_one_epsXzero}
        \Delta(1) = \varepsilon(1)^{-1}\left(1\otimes 1 + \undb^{-1}(s_1X\otimes s_1' X + s_2 X\otimes s_2'X) \right). 
    \end{equation} 
\end{remark}

\begin{remark}\label{rm_symmetries}
What symmetries can we apply to this class of examples, see Remark~\ref{rm_twistings}? One can rescale $X$ by an element of $\mcO^{\ast}$. Also one can change variables $X\mapsto X+w$, so that $uX \mapsto uX +uw$ and $uw\in \mcO$ for all $u\in \mu$. Hence $w\in z^{-1}\mu$. This modifies the trace and takes us outside the case $\varepsilon(X)=0$ since $\varepsilon(X+w)=\varepsilon(w)=w\varepsilon(1)\not=0$ if $w\not=0$.
\end{remark}


\subsection{Case \texorpdfstring{$\varepsilon(X)$}{varepsilonXnotZero} not equal to \texorpdfstring{$0$}{zero}}

Note that $\varepsilon(X)=z^{-1}\undeps_X$ and $\undeps_X\in\mu$. Recall Proposition~\ref{prop_not_equal}. 

\begin{proposition}
    If $\undeps_X\not=0$, then $\undb\not=0$ and  $d\not=0$.
\end{proposition}
\begin{proof} I. If $d=0$, then \eqref{eq_four_two_z} gives $c\undeps_X=0$, so that $\undeps_X=0$ in view of Proposition~\ref{prop_not_equal}, which is a contradiction. 

\vspace{0.05in}


II. Assume $\undb=0$. Then $\undt=z^{-1}\unda \: \undeps_X$ by \eqref{eq_undt}. Inserting this equation into \eqref{eq_four_two_z} produces $(c+d z^{-1}\unda)\undeps_X=0$, so that $c=-d \unda z^{-1}$. Inserting the equation for $c$ in \eqref{eq_four_one_z} gives 
$d(\undeps_X-\unda\varepsilon(1))=z$. 
Since  
both factors $d,\undeps_X-\unda\varepsilon(1)\in \mu$, this gives a contradiction as in Part I in the proof of Proposition~\ref{prop_not_equal}, i.e., the chain
\[
(z)=(d)(\undeps_X- \unda \varepsilon(1)) \subset (d)\mu  \subset   \mu^2 = (z) 
\]
of inclusion of ideals forces a principal and a non-principal ideal to be equal. Thus, $\undb\not=0$.  
\end{proof}

Assume now that $\undeps_X\not=0$. Then $\varepsilon(X)\not=0$ and 
\[
d=(1-c\varepsilon(1))\varepsilon(X)^{-1}, \ \ \ 
c\varepsilon(X)^2+(1-c\varepsilon(1))(a\varepsilon(X)+b\varepsilon(1))=0, 
\]
by \eqref{eq_four_one_corr} for the former and for the latter, substitute in $\undt$ in \eqref{eq_undt} into \eqref{eq_four_two_z}, and then substitute for $c$ in \eqref{eq_four_two_z} using \eqref{eq_four_one_z}.
The latter equation simplifies as 
\begin{equation}
\label{eqn_quadratic_epsX_epsOne}
c\left(\varepsilon(X)^2-a\varepsilon(X)\varepsilon(1)-b\varepsilon(1)^2 \right) = - \left(a\varepsilon(X)+b\varepsilon(1) \right).   
\end{equation}
where the left hand side and the right hand side can be rewritten as 
\begin{equation}
    c\left( \varepsilon(X)^2-\varepsilon(1)\varepsilon(X^2)\right) = -\varepsilon(X^2).
\end{equation}
Inserting this expression for $c$ into the formula for $d$, we obtain
\begin{equation}
    d \left( \varepsilon(X)^2-\varepsilon(1)\varepsilon(X^2)\right) =\varepsilon(X).
\end{equation}
Using this formula and \eqref{eq_four_three_z} gives 
\begin{equation}
    d'\left( \varepsilon(X)^2-\varepsilon(1)\varepsilon(X^2)\right) =-z^{-1}\varepsilon(1).
\end{equation} 
We can write the above three equations as
\begin{equation}
    c\,\wtdelta = \varepsilon(X^2), \ \ \ 
    d\,\wtdelta = - \varepsilon(X), \ \ \ 
    d'\, \wtdelta = z^{-1}\varepsilon(1), 
\end{equation}
cf. \eqref{eq_three_conditions}. 
What are the integrality properties of $\varepsilon(X)^2-\varepsilon(1)\varepsilon(X^2)=-\wtdelta $? We have  
\[\varepsilon(X)\in z^{-1}\mu, \  \varepsilon(1)\in \ \mcO, \ \ \varepsilon(X^2)= t\in z^{-1}\mcO.
\]
Consequently, $\wtdelta\in z^{-1}\mcO$. 

\vspace{0.1in} 

\vspace{0.1in}

Let us further specialize to particular values of $\varepsilon(X)$. 


\subsubsection{Case \texorpdfstring{$\varepsilon(X)=1$}{varepsilonone}}
\label{subsec_case_one}

We have $\varepsilon(X)\in z^{-1}\mu$. 
Let us try $\varepsilon(X)=1$ so that $\undeps_X=z$. 
Then $\undt=\unda+\undb\varepsilon(1)$ and 
the relations \eqref{eq_four_one_z}-\eqref{eq_four_three_z} and \eqref{eq_four_four_a} simplify to 
\begin{eqnarray}\label{eq_four_one_simp_r}
   d & = & 1 - c\varepsilon(1), \\
   \label{eq_four_two_simp_r}
   c z   & = & - d \undt , \\
   \label{eq_four_three_simp_r}
    d'z & = & - d\: \varepsilon(1),  \\
    \label{eq_four_four_simp_r}
    d'\, \undt  & =  & c \,\varepsilon(1).
\end{eqnarray}
The last relation follows from the previous two.
Let us further assume that $\mu $ is prime, which we can do if $\mcO$ is a ring of integers of a number field. (Recall that the Chebotarev Density Theorem implies that any element in $\ClK$ can be represented by a prime ideal $\mu\in \mcO_K$, in any number field $K$, see~\cite[Exercise 11.2.8]{ME05} and~\cite{Aphelli,CCC,Guyot}). Relation \eqref{eq_four_three_simp_r} shows that $d\varepsilon(1)\in(z)$, with $d\in \mu$, so there are two possilities: (I) $\varepsilon(1)\in \mu$ or (II) $d\in (z)$. Case (I) implies, via \eqref{eq_four_one_simp_r},  that $1\in \mu$, a contradiction. Consequently, $d\in(z)$ and we can write $d=z\undd$, with $\undd\in \mcO$. 

Let us rewrite \eqref{eq_four_one_simp_r}-\eqref{eq_four_three_simp_r} and the definition of $\undt$ with the parameter $\undd$:
\begin{eqnarray}\label{eq_one_simp_e}
   z\undd & = & 1 - c\varepsilon(1), \\
   \label{eq_two_simp_e}
   c   & = & - \undd\,\undt , \\
   \label{eq_three_simp_e}
    d' & = & - \,\undd\, \varepsilon(1),  \\
    \label{eq_four_simp_e}
     \undt  & =  & \unda + \undb \varepsilon(1).
\end{eqnarray}
Equations \eqref{eq_two_simp_e} and  \eqref{eq_three_simp_e} allow to get rid of  parameters $c$ and $d'$, reducing the relations to 
\begin{eqnarray}\label{eq_one_simp_f}
   1 & = & \undd (z - \undt \varepsilon(1)), \\
    \label{eq_three_simp_f}
     \undt  & =  & \unda + \undb \varepsilon(1).
\end{eqnarray}
Relation \eqref{eq_one_simp_f} implies that $\undd$ and $z - \undt \varepsilon(1)$ are in $\mcO^{\times}$. Using the expression for $\undt$ equation \eqref{eq_one_simp_f} changes to 
\begin{equation}\label{eq_quadratic}
      \undb\varepsilon(1)^2+ \unda \varepsilon(1) -z = -\undd^{-1}.
\end{equation}
This is a single equation left, and 
our integrality conditions reduce to the following:
\begin{center}
\begin{tabular}{ |c|c|c| } 
\hline
\multicolumn{3}{|c|}{Integrality conditions for $\varepsilon(X)=1$} \\ 
\hline 
 $\mcO$ & $\mu$ & $\mcO^{\times}$\\
\hline
\hline 
  $\undb,\varepsilon(1)$  & $\unda$ & $\undd$\\ 
\hline
\end{tabular}.
\end{center}
Any solution to \eqref{eq_quadratic} with the above integrality conditions gives us an example of a Frobenius algebra over $\mcO$ with desired properties. Condition that $\mu$ be prime is not necessary (but we used it to derive the above equation). 

\vspace{0.07in} 

Equation \eqref{eq_quadratic} can also be written as 
\begin{equation}\label{eq_z_undt}
(z- \undt \varepsilon(1))\undd = 1.
\end{equation}

\vspace{0.07in} 

For more concrete examples, let us further assume that $\varepsilon(1)\in\mcO^{\times}$. Then we can solve \eqref{eq_quadratic} for $\undb$:
\begin{equation}\label{eq_quad_undb}
\undb=\varepsilon(1)^{-2} \left( z-\unda \varepsilon(1)  -\undd^{-1} \right) .
\end{equation}
This gives a family of solutions, parametrized by $\unda\in\mu$ and $\varepsilon(1),\undd\in\mcO^{\times}$: 
\begin{itemize}
    \item $z$ is already given, with $\mu^2=(z)$, 
    \item $X^2=z^{-1}\unda X+z^{-1}\undb$, with $\undb$ given by \eqref{eq_quad_undb}, 
    \item $\varepsilon(1)\in\mcO^{\times}$ is one of the parameters, $\varepsilon(X)=1$, 
    \item $c=-\undd\, \undt$, $\undt=\unda+\undb\varepsilon(1)$ and $d=z\undd=z(z-\undt\varepsilon(1))^{-1}$ satisfy $\wvareps(c+dX)=1^{\ast}$, 
    \item $c'=z^{-1}d$ and $d'=-d\varepsilon(1)$ satisfy $\wvareps(c'+d'X)=z^{-1}X^{\ast}$.
\end{itemize}


\subsubsection{Case \texorpdfstring{$\varepsilon(X)$}{varepsilon(X)} is invertible in \texorpdfstring{$\mcO^{\times}$}{Ostar}}

$\quad$ 

Condition $\varepsilon(X)=1$ in Section~\ref{subsec_case_one} above can be rescaled to cover the case when $\varepsilon(X)\in\mcO^{\times}$. If $\varepsilon(X)=\lambda\in\mcO^{\times}$, define the new trace $\varepsilon_{\mathsf{new}}:=\lambda^{-1}\varepsilon$. Likewise, rescale comultiplication to $\Delta_{\mathsf{new}}=\lambda\Delta$, but keep the multiplication in $A$ the same. This modifies the Frobenius algebra structure of $A$ and makes it satisfy $\varepsilon_{\mathsf{new}}(X)=1$, without changing the $\mcO$-module decomposition $A=\mcO 1 \oplus \mu X$.

Alternatively, one can rescale the generator of $A$ and define $X_{\mathsf{new}}:=\lambda^{-1} X$ without changing $\varepsilon$ and $\Delta$. This changes multiplication in $A$ to 
\begin{equation}
    \label{eq_mult_A_lambda}
    u_1 X_{\mathsf{new}} \, u_2 X_{\mathsf{new}} = \frac{u_1 u_2}{z}\,(\lambda^{-1}\unda X_{\mathsf{new}} +\lambda^{-2}\undb), \ \ \  u_1,u_2\in \mu, 
\end{equation}
cf. \eqref{eq_mult_A}. 


\subsection{Other values of \texorpdfstring{$\varepsilon(X)$}{epsilon(X)}}

The integrality condition says that $\varepsilon(X)\in z^{-1}\mu$. 
So far we have found examples with $\varepsilon(X)=0$ and $\varepsilon(X)\in \mcO^{\times}$. Can we find examples where $\varepsilon(X)\not=0$ and $\varepsilon(X)$ is not an invertible element of $\mcO$? 

Consider one of the simplest examples of the ring of integers $\mcO=\Z[\sqrt{-5}]$, with $K=\Q(\sqrt{-5})$ and a nontrival ideal class group $\Cl(\mcO)\cong \Z/2$, see~\cite{Conrad19_ideal_classes}. The ideal $\mu=(2,1+\sqrt{-5})$ in $\mcO$ is non-principal, and $\mu^2=(2)$, so we set $z=2$. The quotient $\mcO/\mu\cong \FF_2$, so $\mu$ is prime in $\mcO$. The only invertible elements in $\mcO$ are $\pm 1$. Let us try 
\[
\varepsilon(X) = \frac{1+\sqrt{-5}}{2}\in 2^{-1}\mu 
\]
so that $\undeps_X= 1+\sqrt{-5}$. We need to find a system of parameters so that relations \eqref{eq_four_one_z}-\eqref{eq_four_three_z} and \eqref{eq_four_four_a} hold, subject to the integrality conditions in the table preceeding equation \eqref{eq_undt}. 
Looking at equations \eqref{eq_four_two_z}, \eqref{eq_four_three_z}, specialized to our case: 
\[
c\:(1+\sqrt{-5})  =  - d \:\undt , \ \ 
   d'(1+\sqrt{-5}) =  - d\: \varepsilon(1),
\]
we can set $d=s(1+\sqrt{-5})$ for $s\in\{\pm 1\}$ and reduce the two equations to 
\[
c=-s\undt, \ \ d'=-s\varepsilon(1), 
\]
allowing us to replace $c$ and $d'$ in the other two equations. Equation \eqref{eq_four_four_a} becomes trivial, and Equations~\eqref{eq_undt} and \eqref{eq_four_four_z} are reduced to a formula and an equation
\begin{eqnarray*}
    \undt & = & 2^{-1}\unda(1+\sqrt{-5})+\undb \varepsilon(1),    \\
    1 & = & -s \undt \varepsilon(1) + s (\sqrt{-5}-2), 
\end{eqnarray*}
respectively. Inserting $\undt$ into the 2nd equation above, one is left with 
\[
(2^{-1}\unda(1+\sqrt{-5})+\undb \varepsilon(1))\varepsilon(1) = \sqrt{-5}-2-s.
\]
Specialize to $\unda=1-\sqrt{-5}\in\mu$  to get 
\[
(3 + \undb \varepsilon(1))\varepsilon(1)= \sqrt{-5}-s-2 . 
\]
Further specializing to $\varepsilon(1)\in\{\pm 1\}\subset \mcO$ results in 
\[
\undb = \varepsilon(1) \left( (\sqrt{-5}-s-2)\varepsilon(1) - 3 \right)  \in \mcO. 
\]
One can check that all integrality conditions are satisfied by this solution, giving us an example of $A$ as above with $\varepsilon(X)$ neither $0$ nor in $\mcO^{\times}$ (and not in $\mcO$). We were quite loose with our choices, and there should be many examples with $\varepsilon(X)$ neither $0$ nor an invertible element of $\mcO$. 

One possible family of such examples would be first to set $d=s \undeps_X$ for $s\in \mcO^{\times}$, simplifying equations 
\eqref{eq_four_two_z} and \eqref{eq_four_three_z} to formulas for $c$ and $d'$ and making \eqref{eq_four_four_a} an identity, substitute in $\undt$ in \eqref{eq_undt} and then look for various solutions of the remaining single equation
\begin{equation}
-s ( z^{-1}\unda\: \undeps_X+\undb\varepsilon(1)) \:\varepsilon(1)+z^{-1}s\undeps_X^2 = 1,
\end{equation}
subject to the integrality conditions.


\section{Towards link homology}

Consider the tensor product $A^{\otimes 2}=A\otimes_{\mcO}A$ and denote by $\ell_x$ the operator of multiplication by $x\in A$ on the first term, 
\begin{equation}
\label{eqn_ell_x}
    \ell_x(y_1\otimes y_2)=xy_1\otimes y_2. 
\end{equation}
Multiplication by elements of $A$ on the first term turns $A^{\otimes 2}$ into a projective $A$-module isomorphic to 
\[
A \otimes A = (A\otimes 1) \oplus (A\otimes \mu X) \cong A \oplus (A\otimes_{\mcO}\mu). 
\]
Indeed, tensoring with $\mu$ over $\mcO$ is an involution on the category of $A$-modules, taking projectives to projectives. 
It would be interesting to understand under what conditions there is an isomorphism of $A$-modules $A\cong A\otimes_{\mcO}\mu$. 

Consider the multiplication map $m:A\otimes A\lra A$ and 
recall the formula \eqref{eq_mult_A} for it: 
\begin{equation}
    \label{eq_mult_A_again}
    u_1 X \, u_2 X = \frac{u_1 u_2}{z}\,(\unda X +\undb), \ \ \  u_1,u_2\in \mu. 
\end{equation}

There are isomorphisms $\mu\otimes \mu \cong (z)\cong \mcO$ of $\mcO$-modules. Pick $k\ge 2$ and elements $u_1,\dots, u_k\in\mu$ that generate the ideal $\mu$. Any ideal in a Dedekind domain can be generated by two elements, so one can always choose $k=2$ above. Pick elements $u_1',\dots, u_k'\in \mu$ so that 
\begin{equation}\label{sum_ujs_z}
\sum_{j=1}^k u_ju_j'=z. 
\end{equation}
We use shorthand $u_j\otimes u_j'$ to denote the sum $\sum_{j=1}^k u_j\otimes u_j'$ and likewise for $u_ju_j'=z$. 

Introduce the following elements 
\begin{equation}
\begin{split}
    X_u & \ := \ uX \otimes 1 - 1\otimes u X \in A^{\otimes 2}, \ \ u\in \mu,\\
   \label{eqn_suppressed} 
   \widehat{X} & \ := \  u_jX\otimes u_j' X - (\unda X\otimes 1+ \undb 1\otimes 1) \in A^{\otimes 2}, \\
\end{split}
\end{equation}
where in the second formula in~\eqref{eqn_suppressed}, the sum sign over the index $j$ is suppressed. 
Note that
\begin{equation}\label{eq_us_Xs}
u_jX u_j' X =  (u_ju_j') X^2  = z(z^{-1}\unda X+z^{-1}\undb  1) = \unda X +\undb 1,
\end{equation}
where the second equality in \eqref{eq_us_Xs} holds due to \eqref{eq_mult} and the first two rescaling of the parameters in \eqref{eq_param_shift_new}. 
Then 
\[X_{u_1}+X_{u_2}=X_{u_1+u_2}, \ \ \ 
s X_u = X_{su}, \ s\in \mcO, 
\]
and 
\begin{equation}\label{eq_whX_diff}
    \widehat{X} = u_jX\otimes u_j' X - u_jX u_j' X\otimes 1, 
\end{equation}
in view of \eqref{eq_us_Xs}. 
The set $X_{\mu}:=\{X_u\}_{u\in \mu}$ is an $\mcO$-submodule of $A^{\otimes 2}$. 
\begin{prop}
    We have 
    \begin{equation}
        \ker(m) = X_{\mu}\oplus \mcO \widehat{X}.
    \end{equation}
\end{prop}
\begin{proof} Formulas \eqref{eq_us_Xs} and \eqref{eq_whX_diff} imply that $\widehat{X}\in \ker(m)$, and it is clear that $X_\mu\oplus \mcO \widehat{X}\subset \ker(m)$. 
    Vice versa, given an element $y\in A^{\otimes 2}$, by adding the term $X_u$ for some $u$, we can assume that $y$ does not contain a multiple of $1\otimes X$. Due to \eqref{eq_whX_diff} and subtracting an $\mcO$-multiple of $\widehat{X}$, we can assume that $y$ does not contain a multiple of $X\otimes X$ and reduce it to the form 
    \[
    y =  y_1 X\otimes 1+ y_2 1\otimes 1, 
    \]
    for some $y_1\in \mu$ and $y_2\in\mcO$. The condition $y\in \ker(m)$ forces $y_1=y_2=0$.
\end{proof}
Recall~\eqref{eqn_ell_x}. Using~\eqref{sum_ujs_z}, we have 
\begin{equation}\label{eq_act_ell}
    \ell_{uX}(\widehat{X}) = - X_{\undb u} +\frac{u \unda}{z}\widehat{X}, \ \ \ 
    \ell_{uX}(X_{u'}) = - \frac{u u'}{z} \widehat{X}. 
\end{equation}
Consequently, there is an isomorphism of $\mcO$-modules $A\otimes_{\mcO} \mu\cong \ker(m)$ taking $1\otimes u$ to $X_u$. 

\vspace{0.07in} 

To have an isomorphism of $A$-modules $\ker(m)\cong A$ requires $\ker(m)$ to have a single generator as an $A$-module. Such a generator must have the form 
\begin{equation}
    \widetilde{X}= \gamma \widehat{X} - X_u, 
\end{equation}
for some $\gamma\in \mcO$ and $u\in \mu$. Using \eqref{eq_act_ell} we have 
\begin{equation}
    \ell_{u'X}(\widetilde{X}) = \frac{u u'}{z}\widehat{X} + \gamma\left( -X_{\undb u'} + \frac{u' \unda}{z}\widehat{X}\right) 
    = \frac{u'}{z}(\gamma \unda + u)\widehat{X}-  X_{\gamma \undb u'}. 
\end{equation}
A necessary condition for $\widehat{X}$ to generate $\ker(m)$ is that $\mu = \mcO u +  \gamma \undb\mu$ and that $\mcO$ is generated by $\gamma$ and elements $\frac{u'}{z}(\gamma \unda + u)$ over $u'\in \mu$.

Let us specialize to $\gamma=1$ (this also covers the case $\gamma\in \mcO^{\times}$). Then 
\begin{equation}
    \widetilde{X}=\widehat{X} - X_u,  \  \ 
    \ell_{u'X}(\widetilde{X}) = \frac{u'}{z}(\unda + u)\widehat{X}-  X_{\undb u'}, 
\end{equation}
and 
\begin{equation}
    \ell_{u'X}(\widetilde{X}) -\frac{u'}{z}(\unda + u)\widetilde{X} =  X_{u'(-\undb + \frac{u}{z}(\unda+u))}.
\end{equation}
Consequently, a sufficient condition for $\ker(m)\cong A$ is that for some $u\in \mu$ the element 
\begin{equation}\label{eq_invert_element}
-\undb +  \frac{u(\unda+u) }{z}
\end{equation}
is invertible in $\mcO$. Taking $u=0$, we see that $\widehat{X}$ is a generator of $\ker(m)$ if and only if $\undb$ is invertible. The Frobenius algebras with $\varepsilon(X)=0$ described in Section~\ref{subsec_caseone} have this property. 
This at least gives us some examples when $\ker(m)\cong A$ as an $A$-module.

We do not know a classification or a convenient description of Frobenius algebras as above for which there is an isomorphism $\ker(m)\cong A$.

\vspace{0.07in} 

A Frobenius algebra $A\cong R[X]/(X^2-aX-b)$ over the ground ring $R$, of rank two with a basis $\{1,X\}$ and a fixed Frobenius structure gives a link homology theory~\cite{Kho_link_06,Bar_Natan05,KhovRob_link_hom_II}. Starting with a planar diagram $D$ of a link $L$, one can repeat the construction of~\cite{Kho_Jones00}. One resolves $D$ into $2^k$ planar diagrams, where $k$ is the number of crossings of $D$. To a diagram with $n$ circles, one associates $A^{\otimes n}$. These $R$-modules fit into a commutative cube via maps given by the multiplication $m$ and comultiplication $\Delta$ in $A$ on proper terms in the tensor product. The commutative cube is collapsed into the complex $C(D)$, with an appropriate grading shift. Homology $H(D)=H(C(D))$ of the complex can be shown to be invariant under the Reidemeister moves I, II, III, resulting in a link homology theory $H(L)$ for oriented links $L$.  

\begin{figure}
    \centering
\begin{tikzpicture}[scale=0.6,decoration={
    markings,
    mark=at position 0.50 with {\arrow{>}}}]
\begin{scope}[shift={(0,0)}]

\draw[thick] (0,0) -- (0,4);

\node at (1,2.5) {RI};

\node at (-0.6,3.3) {$D$};
\node at (2.7,3.5) {$D_1$};

\node at (1,2) {$\sim$};

\draw[thick] (2,0) .. controls (2,3) and (4,3) .. (4,2);

\draw[thick] (2,4) .. controls (2,3) and (2.15,2.65) .. (2.35,2.15);

\draw[thick] (2.5,1.85) .. controls (2.75,1) and (4,1) .. (4,2);
\end{scope}

\begin{scope}[shift={(9,0)}]

\draw[thick] (0,0) -- (0,4);

\draw[thick] (3,2) arc (0:360:1);

\draw[thick,->] (4,2) -- (5.5,2);

\end{scope}

\begin{scope}[shift={(15,0)}]

\draw[thick] (0,4) .. controls (0.75,1.25) and (1.75,3) .. (2,3);

\draw[thick] (0,0) .. controls (0.75,2.75) and (1.75,1) .. (2,1);

\draw[thick] (2,1) .. controls (3.75,0.15) and (3.75,3.85) .. (2,3);

\end{scope}

\end{tikzpicture}
    \caption{Left: One of the two flavors of the Reidemeister I move of link diagrams. Middle: two resolutions of the curl diagram $D_1$ on the right hand side of the RI move. Right: Depicting the tensor product $A\otimes \mu$ by a dot on a strand; two adjacent dots can be removed due to the isomorphism $\mu\otimes_{\mcO}\mu\cong \mcO$. }
    \label{fig_00001}
\end{figure}
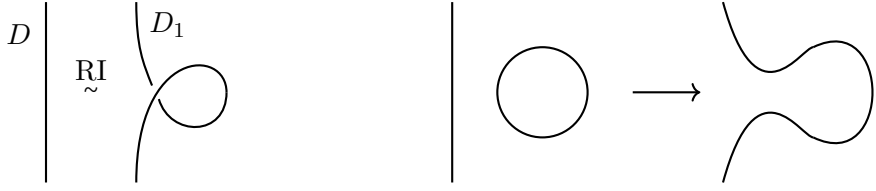

The key property of the Frobenius algebra $A$ is that it is of rank 2 over $R$. This is necessary for the invariance under the Reidemeister I move, shown in Figure~\ref{fig_00001} on the left. In that move, one adds a curl to an interval of the diagram $D$ resulting in new diagram $D_1$. For one of the two possible types of the curl, the complex assigned to $D_1$ at the crossing is 
\begin{equation}\label{eq_comp_R1}
0\lra A\otimes A \stackrel{m}{\lra} A \lra 0. 
\end{equation}
It breaks down into the sum of a contractible complex $0 \lra A\stackrel{1}{\lra}A\lra 0$ and the complex $0\lra A \lra 0$. The latter complex is also what the interval contributes to $C(D)$. Removing the contractible summands leads to a homotopy equivalence $C(D_1)\cong C(D)$ and an isomorphism of homology groups $H(D_1)\cong H(D)$. 
Reidemeister I move for the opposite type of crossing is dealt with by passing to the dual of the complex \eqref{eq_comp_R1}, via the isomorphism $\wvareps:A\cong A^{\ast}$, and simplifying it in the same manner. 

If $R$ and $A$ are additionally $\Z$-graded, and Frobenius structure respects the grading, then $H(L)$ is bigraded. Otherwise $H(L)$ has a single grading. 

Using a cyclic Frobenius algebra $A$ of higher rank over $R$, with a higher degree relation $X^n=f(X)$, for $n>2$ and $\deg(f)\le n-1$,  does not work, since removing the contractible summand leaves one with the complex $0 \lra (A/R\cdot 1)\otimes_{R} A\lra 0$ isomorphic to $n-1$ copies of $0\lra A\lra 0$, since $\rk_R(A/R\cdot 1)=n-1$. The higher rank cyclic extensions do give rise to link homology, but that requires working with matrix factorizations or foam evaluation or bringing in suitable categories to make it work. Some of these approaches can be found in papers~\cite{KR_matrixfact08,RW_eval_foams20,CK_II08,Man07,Web_higher17,RW16}. 

\vspace{0.07in} 

Consider a rank 2 Frobenius extension $(A,\mcO)$ from the present paper where $A$ is a projective but not a free $\mcO$-module. Given a planar diagram $D$ of a link $L$, one can form the complex $C(D)$ as above, using the tensor powers of $A$ over $\mcO$. The complex is singly-graded (no $q$-grading) and has homology groups $H(D)$. 

Tensoring the complex with the field of fractions $K$ results in the complex $C(D)\otimes_{\mcO}K$ of $K$-vector spaces whose homology groups $H(D)\otimes K$ are invariant under the Reidemeister moves and give a version of the Eun Soo Lee link homology theory~\cite{Lee_endo05}, due to $X^2=aX+b$ having two distinct roots in some quadratic extension of $K$. 

It is natural to ask whether, for some of these Frobenius extensions, groups $H(D)$ themselves are link invariants. The easiest way for this to happen is to have homotopy equivalences $C(D)\cong C(D_1)$ for diagrams $D,D_1$ related by a Reidemeister move. 

Let us examine Reidemeister move I, already discussed above.  The complex \eqref{eq_comp_R1} breaks down into the sum of a contractible complex and the complex $0\lra \ker(m)\lra 0$. The latter is isomorphic to $A\otimes \mu$. We would like it to be isomorphic to $A$, see the earlier discussion in this section. In particular, it is isomorphic to $A$ if $\varepsilon(X)=0$, but that is a rather special condition.  

It would be useful to further develop this construction, and 
\begin{itemize}
    \item understand under what conditions there is an isomorphism $\ker(m)\cong A$, 
    \item assuming $\ker(m)\cong A$, check what additional conditions are required for the invariance of the complex $C(D)$ under Reidemeister moves II, III and look for examples of such Frobenius $\mcO$-algebras $A$.  
\end{itemize}
 
\begin{remark} Since $\ker(m)\cong A\otimes \mu$, instead of looking for an isomorphism $A\otimes \mu\cong A$, another possibility is to depict the  tensor product $A\otimes \mu$ by a dot on a strand and replace the Reidemeister I move by the move RI' shown in Figure~\ref{fig_00002}. Two dots on a strand encode the tensor product with $\mu\otimes \mu\cong \mcO$ and could be erased. More precisely, 
one sets us cobordisms between a strand with no dots and a strand with two dots that are mutually-inverse module isomorphisms $\mu\otimes \mu\leftrightarrow \mcO$. 

\vspace{0.1in}

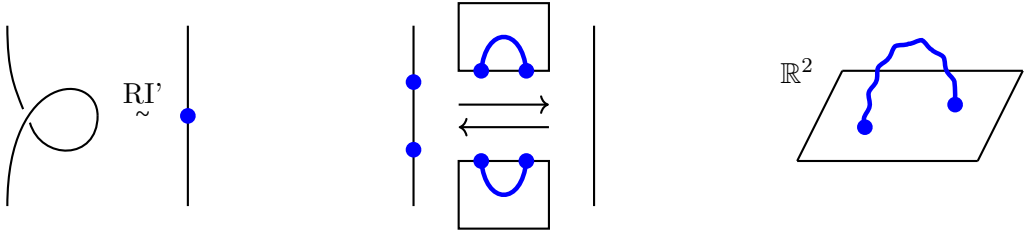
\begin{figure}
    \centering
\begin{tikzpicture}[scale=0.6,decoration={
    markings,
    mark=at position 0.50 with {\arrow{>}}}]
\begin{scope}[shift={(1,0)}]

\draw[thick] (0,0) .. controls (0,3) and (2,3) .. (2,2);

\draw[thick] (0,4) .. controls (0,3) and (0.15,2.65) .. (0.35,2.15);

\draw[thick] (0.5,1.85) .. controls (0.75,1) and (2,1) .. (2,2);

\draw[thick] (4,0) -- (4,4);

\node at (3,2.5) {RI'};

\node at (3,2) {$\sim$};

\draw[thick,fill,blue] (4.15,2) arc (0:360:1.5mm);
\end{scope}




\begin{scope}[shift={(10,0)}]

\draw[thick] (0,0) -- (0,4);

\draw[thick,fill,blue] (0.15,2.75) arc (0:360:1.5mm);

\draw[thick,fill,blue] (0.15,1.25) arc (0:360:1.5mm);

\draw[thick,->] (1,2.25) -- (3,2.25);

\draw[thick,<-] (1,1.75) -- (3,1.75);

\draw[thick] (4,0) -- (4,4);

\draw[thick] (1,3) rectangle (3,4.5);

\draw[thick,fill,blue] (1.65,3) arc (0:360:1.5mm);

\draw[thick,fill,blue] (2.65,3) arc (0:360:1.5mm);

\draw[line width=0.6mm,blue] (1.5,3) .. controls (1.6,4) and (2.4,4) .. (2.5,3);

\draw[thick] (1,1) rectangle (3,-0.5);

\draw[thick,fill,blue] (1.65,1) arc (0:360:1.5mm);

\draw[thick,fill,blue] (2.65,1) arc (0:360:1.5mm);

\draw[line width=0.6mm,blue] (1.5,1) .. controls (1.6,0) and (2.4,0) .. (2.5,1);
\end{scope}

\begin{scope}[shift={(18.5,1)}]

\draw[thick] (1,2) -- (5,2);

\draw[thick] (0,0) -- (4,0);

\draw[thick] (0,0) -- (1,2);

\draw[thick] (4,0) -- (5,2);

\node at (0,2) {$\mathbb{R}^2$};

\draw[thick,fill,blue] (3.65,1.25) arc (0:360:1.5mm);

\draw[thick,fill,blue] (1.65,0.75) arc (0:360:1.5mm);

\draw[thick,line width=0.6mm,blue] (1.5,0.75) decorate [decoration={snake,amplitude=0.25mm}] {.. controls (1.75,3) and (3.25,3) .. (3.5,1.25)};

\end{scope}

\end{tikzpicture}
    \caption{Left: RI' move, a variation on the Reidemeister I move. The tensor product $A\otimes \mu$ is depicted by a dot on a strand. Middle: two adjacent dots on a strand can be canceled due to an isomorphism $\mu\otimes_{\mcO}\mu\cong \mcO$. Right: dots can be moved out of the strands to the plane, with the canceling cobordism on a pair of dots shown. }
    \label{fig_00002}
\end{figure}


To a collection of $k$ planar circles, one attaches $A^{\otimes k}$. Whichever copy of $A$ in this tensor product is tensored with $\mu$ does not make a difference for the tensor product. Consequently, the dot representing $\mu$ does not have to be on any particular edge or circle and can be placed anywhere in the plane. A pair of dots in the plane can be canceled, which can be represented by an arc in $\R^2\times [0,\infty)$ connecting them. This arc can intersect facets of a cobordism between 1-manifolds in $\R^2$ without interacting with the facets.

In this setup, however, there is a problem with the Reidemeister II move, where, in the resolution of the 2-crossing diagram, two of the resolutions $D_{00},D_{11}$ are isotopic and the third resolution, say $D_{10}$, is isotopic to each of them with an additional circle. This circle contributes $A$ to the state space of diagram  $D_{10}$. The summand $\mcO 1$ can be canceled with the state space of $D_{11}$, but for the remaining term $\mu X$ to cancel with $D_{00}$ via a contractible summand one still expects an isomorphism $A\otimes \mu\cong A$. 
\end{remark}

\begin{remark} Our examples 
of rank two Frobenius $\mcO$-algebras that are not free over $\mcO$ are not graded. It is easy to produce examples of rank $N>2$ commutative graded Frobenius $\mcO$-algebras. For this, assume that $\mu$ is a non-principal ideal in $\mcO$ of order $N-1$ in $\Cl(\mcO)$. Let $\mu^{N-1}=(z)$, so that an isomorphism $\mu^{\otimes {N-1}}\cong (z)\cong \mcO$ is fixed. 

Start with the Frobenius algebra $A'_N=\mcO[X]/(X^N)$ with the trace $\varepsilon(X^{N-1})=z^{-1}$ and
 $\varepsilon(X^m)=0$ for $m<N-1$. The comultiplication is given by 
\[\Delta(1)=z(X^{N-1}\otimes 1 + X^{N-2}\otimes X+\ldots + 1\otimes X^{N-1}).
\]
Tensoring with $K$ produces a commutative Frobenius $K$-algebra $A'_N\otimes K$. Set $\deg X =2$ to make $A'_N$ and $A'_N\otimes K$ into graded Frobenius algebras. 

Consider $\mcO$-subalgebra of  $A'_N\otimes K$ given by 
\[
A_N:= \mcO1 \oplus \mu X\oplus \mu^2 X^2\oplus \ldots \oplus \mu^{N-1}X^{N-1}.  
\]
Due to the isomorphism $\mu^{N-1}\cong \mcO$, the trace $\varepsilon$ restricts to an isomorphism $\mu^{N-1}X^{N-1}\cong\mcO$. It is straightforward to see that $A_N$ is a Frobenius $\mcO$-algebra. 

The algebras $A'_N\otimes K$, when $N$ is odd and $\sqrt{N}\in K$, give examples of Turaev--Turner unoriented 2D TQFTs~\cite{TT06}. Algebra $A'_3$ is used in the construction of $\mathfrak{sl}(3)$-link homology~\cite{Kh_sl304}. It is possible that some twists of this algebra can give rise to modifications of $\mathfrak{sl}(3)$-homology, and likewise for the $\mathfrak{sl}(N)$-link homology. 
\end{remark}

\begin{remark} An isomorphism of $\mcO$-modules $A\cong A^{\ast}$, see Section~\ref{section:intro}, constraints us to elements $P$ of order two in $\Cl(\mcO)$ as a summand of $A$. This isomorphism comes from a diffeomorphism between an oriented circle $\SS^1$ and the same circle with the opposite orientation. Going one dimension up, if looking for an oriented 3D TQFT over $\mcO$, one would likewise be forced into elements of order 2 in the ideal class group as summands of $\mcF(M)$, for an oriented surface $M$, due to diffeomorphisms $M\cong (-M)$. 

Returning to dimension two, one way to avoid this restriction is to add additional decorations to 1-manifolds and 2-cobordisms between them. For example, one can pick a group $G$ and consider 2-cobordisms $S$ equipped with a principal $G$-bundle over them. A oriented circle $\SS^1$ is then replaced by a pair $(\SS^1,c)$, where $c$ is a conjugacy class in $G$. Reversing the orientation results in the pair $(\SS^1,c^{-1})$, which is diffeomorphic to $(\SS^1,c)$ only if the conjugacy classes $c$ and $c^{-1}$ coincide. It would be interesting to investigate such decorated 2D TQFTs over Dedekind domains $\mcO$, including the special case $G=\Cl(\mcO)$. 
\end{remark}

\bibliographystyle{amsalpha} 
\bibliography{alg_fields}

\end{document}